\documentclass[12pt,review]{elsarticle}

\usepackage[margin=2.2cm]{geometry}
\usepackage{mathrsfs}
\usepackage{amssymb}
\usepackage{amsbsy}
\usepackage{amsmath}
\usepackage[mathscr]{euscript}
\usepackage{amssymb}
\newtheorem{claim}{Claim}
\newproof{proof}{Proof}
\newdefinition{definition}{Definition}
\newdefinition{example}{Example}
\journal{arXiv}

\begin{document}

\begin{frontmatter}



\title{Networks of planar Hamiltonian systems}


\author{David S. Tourigny}
\ead{dst27@cam.ac.uk}
\address{Department of Applied Mathematics \& Theoretical Physics \\ University of Cambridge \\ Wilberforce Road \\ Cambridge CB3 0WA, United Kingdom}

\begin{abstract}
We introduce diffusively coupled networks where the dynamical system at each vertex is planar Hamiltonian. The problems we address are synchronisation and an analogue of diffusion-driven Turing instability for time-dependent homogeneous states. As a consequence of the underlying Hamiltonian structure there exist unusual behaviours compared with networks of coupled limit cycle oscillators or activator-inhibitor systems.  
\end{abstract}


\end{frontmatter}


\section{Introduction}

In their modern-day form, dynamical systems on complex networks have captured the interest of researchers for several decades \cite{Str01,Boc06,Bar08}. Applications span a wide variety of disciplines ranging from the social, biological and neurological sciences, all the way to computing, engineering and the structure of the World Wide Web. Although the study of complex networks can mean different things to different scientists, one of most prevalent concepts is that a network encompasses a set of rules for joining together many autonomous units to produce a high-dimensional dynamical system. In practice, each individual unit is itself a dynamical system of dimension much smaller than the network as a whole, and the precise way these are coupled together is encoded by a combinatorial graph. We shall use the term ``complex network'' when referring to a dynamical system constructed in this way, but it can also be used to describe the combinatorics itself, e.g. ``a dynamical system {\em on} a complex network''. Within this class of dynamical systems are complex networks whose autonomous units admit oscillatory solutions. Arrays of coupled limit cycle oscillators \cite{Win67,Kur84, Are08} are important examples.

Coupled limit cycle oscillators have been used to study phenomena associated with temporal periodic behaviour, such as synchronisation \cite{Are08}. Beyond the various forms of synchronisation there also exist notions of amplitude and oscillation death \cite{Kos13}. The latter can be viewed as an oscillatory analogue of Turing's activator-inhibitor model for pattern formation on complex networks \cite{Tur52,Nak10}. A basic assumption in each of these cases is that all autonomous units of the network admit attractors. Namely, asymptotically stable limit cycles or asymptotically stable equilibria for synchronisation or pattern formation, respectively. Limit cycles allow the dynamics of synchronisation to be well-approximated by coupled phase oscillators \cite{Win67,Kur84} whereas stable equilibria enable one to define an unpatterned state of a network \cite{Tur52,Nak10}. For complex networks whose autonomous units do not admit attractors, notions of synchronisation and pattern formation must be modified appropriately. Temporal behaviour of these complex networks may be very different from that of their canonical counterparts.  

Synonymous with time evolution in classical mechanics are dynamical systems of Hamiltonian form. Hamiltonian systems famously do not admit attractors, but rather families of periodic orbits parameterised by level sets of the Hamiltonian. There have been surprisingly few attempts to study complex networks consisting of coupled Hamiltonian systems. This is perhaps due to an absence of attractors or the applications of coupled limit cycle oscillators being so prolific. Consequently, in this paper we initiate the study of complex networks whose autonomous units are planar Hamiltonian systems. We choose to focus on planar Hamiltonian systems because already there exists a rich theory that may be identified with the geometry of planar algebraic curves. After introducing the general model we begin to investigate synchronisation and pattern formation in this new class of complex networks.                           

\section{General theory}

\subsection{Planar Hamiltonian systems}
Planar systems are described by $C^1$-vector fields on the plane $\mathbb{R}^2$. In coordinates $(x,y)$ they take the form
\begin{equation} \label{plane} \dot{x} = f(x,y) , \quad \dot{y} = g(x,y) , \end{equation}
where $f,g$ are $C^1$-functions of $(x,y)$. A special class of planar systems are described by those vector fields whose divergence vanishes identically and can therefore give rise to periodic orbits in any region of the plane. Clearly this is satisfied by taking $f=\partial h/ \partial y$ and $g=-\partial h/\partial x$, where $h$ is a smooth (at least $C^2$) function $h: \mathbb{R}^2 \to \mathbb{R}$, and so we are led to {\em planar Hamiltonian systems}
\begin{equation} \label{ham} \dot{x} = \frac{\partial h(x,y) }{\partial y} , \quad \dot{y} =- \frac{\partial h (x,y)}{\partial x}. \end{equation}
Planar hamiltonian systems and their $2N$-dimensional relatives have a rich associated theory. A simple consequence of their form is that the Hamiltonian $h$ is preserved in time meaning it defines an {\em integral of motion} (it is easy to check that $\dot{h} = 0$ using the chain rule). If a $2N$-dimensional Hamiltonian admits $N$ linearly-independent integrals of motion then the system it defines is said to be {\em integrable} (in the Liouville sense). Planar Hamiltonian systems are unique in the sense that the Hamiltonian always provides the required $N=1$ integral of motion and so {\em every autonomous planar Hamiltonian system is also integrable}.    

Trajectories of the system (\ref{ham}) can be parameterised by the value of $h(x,y)$ and periodic solutions correspond to closed components of the curve $h(x,y) = \mbox{const}$. This combination of integrability and periodicity implies that there exists a special coordinate transformation to {\em action-angle variables} $ (x,y) \to (I, \theta)$ such that $h(x,y) = h(I)$, i.e., in these coordinates the Hamiltonian depends only on one variable, $I$. The equations transform as
\[ \dot{I} = \frac{\partial h(I) }{\partial \theta} = 0 , \quad \dot{\theta} =- \frac{\partial h (I)}{\partial I} \equiv \omega(I) \]
and one sees that the action variable $I \in \mathbb{R}_{\geq 0 }$ remains constant whilst the angle variable $\theta \in S^1$ evolves according to $\theta(t) = \omega(I)t + \theta(0)$. Roughly this result can be summarised by saying that, if a planar Hamiltonian  system begins on a periodic orbit labelled by $I = I(0)$ at time $t = 0$, then it remains on that periodic orbit for all time where it evolves with the characteristic frequency $\omega(I)$. The fact that trajectories of a planar Hamiltonian vector field are parameterised by curves in the level set of the Hamiltonian is a good indication of the close ties between these systems and planar geometry.

The structure of Hamiltonian systems can sometimes mean that conventional methods for the analysis of generic dynamical systems fail. An important illustration of this arises when attempting to apply the Hartman-Grobman Theorem to evaluate the stability of an equilibrium point. Suppose that the planar system (\ref{ham}) admits an equilibrium solution $(x^*,y^*)$, which we see must be equivalent to the statement that $h$ has at least one critical point. Calculating the Jacobian of the linearisation of (\ref{ham}) around $(x^*,y^*)$ we find that eigenvalues come in pairs
\[ \lambda_{\pm} = \pm \sqrt{-\mbox{det}[\mbox{Hess}(h)(x^*,y^*)]} ,\] 
where $\mbox{Hess}(h)(x^*,y^*)$ is the Hessian of the Hamiltonian evaluated at $(x^*,y^*)$. If the determinant is zero or positive the eigenvalues form a purely imaginary conjugate pair and the Hartman-Grobman Theorem does not apply. In this case the equilibrium is called a {\em centre}. The inability to classify Hamiltonian centres as asymptotically stable fits into a larger scheme that arises as a consequence of a system having Hamiltonian structure, namely that {\em Hamiltonian systems can not have asymptotically stable (attracting) sets}. This peculiar property distinguishes periodic Hamiltonian systems from generic dynamical systems that admit limit cycles, and necessarily implies a centre must be surrounded by a family of periodic orbits. Existence of these families of periodic orbits will turn out to be a major factor distinguishing networks of planar Hamiltonian systems when we turn to synchronisation.    

\subsection{Complex networks}
The term ``network'' has been much bandied in the literature, but here we shall use a self-contained definition. For us a {\em network} is a particular type of dynamical system specified by three ingredients: 

1) a pair $(V,W)$ consisting of a finite set $V=\{1,2,...,n\}$ and a function $W: V \times V \to  \mathbb{R}$;    

2) a set of $n$ autonomous $m$-dimensional dynamical systems $\{\dot{z}_i = F_i(z_i)\}$;

3) a pairwise ``coupling function'' $U: \mathbb{R}^m \times \mathbb{R}^m \to \mathbb{R}^m$.  
\\ 
These data are enough to define a weighted graph $\mathcal{G}$ by assigning  a vertex to each $i \in V$ and an edge with ``weight'' $W_{ij} = W(i,j)$ joining vertex $i$ to $j$. We assume this graph $\mathcal{G}$ is {\em undirected} (i.e., $W_{ij} = W_{ji}$) and {\em connected} meaning that there always exists a path from any vertex $i$ to any other vertex $j$. The network is then defined to be the dynamical system on $\mathbb{R}^{m \times n}$ whose time evolution is governed by the equations
\[ \dot{z}_i = F_i(z_i) + \sum_{j=1}^n W_{ij} U(z_i,z_j) \ . \]  
In this paper we will actually only be concerned with a small subset of networks where the autonomous units are planar ($m = 2$) and the coupling function is {\em diffusive}, implying $U(z_i,z_j) = U(z_i - z_j)$. Finally, with this is mind we say that a network is an {\em oscillator network} if each of the underlying systems $\dot{z_i} = F_i(z_i)$ admits at least one periodic solution.    

Networks of the above type have appeared at multiple points throughout history. The two best-studied examples are networks of limit cycle oscillators and activator-inhibitor networks, which form canonical models for synchronisation and pattern formation, respectively. Briefly, networks of limit cycle oscillators were popularised  by A.T. Winfree \cite{Win67} to study synchronisation of biological oscillations. In Winfree's model each $F_i$ admits a limit cycle solution and coupling is assumed to be ``weak''. His achievement was to show that with this assumption the amplitudes of limit cycles can be neglected and networks reduce to a system of coupled phase oscillators. Later, Kuramoto expanded on Winfree's work and demonstrated that if coupling is approximately diffusive one can find an analytically tractable solution \cite{Kur84,Are08}. In contrast, activator-inhibitor networks with diffusive coupling were introduced by A. Turing to understand a different biological phenomenon \cite{Tur52}. In his model all the underlying planar systems are the same ($F_i = F$ for all $i \in V$) and $U$ is linear, so the modern-day form \cite{Nak10} of Turing's equations is
\[ \dot{x_i} =  f(x_i,y_i) + D_x \sum_{j=1}^n W_{ij} (x_j-x_i) , \quad \dot{y_i} =  g(x_i,y_i) + D_y \sum_{j=1}^n W_{ij} (y_j-y_i) .\]
Turing showed that it is possible for a uniformly distributed equilibrium solution, $(x_i,y_i) = (x^*,y^*)$ for all $i \in V$, to spontaneously destabilise in the presence of nonzero diffusion coefficients $D_x,D_y>0$ and suggested this as a model for pattern formation.    

Networks considered by Winfree and Turing are at opposite ends of a wide spectrum of examples. Those introduced by Winfree are examples of {\em heterogeneous networks} where the autonomous planar system at each vertex is {\em different}. Turing's networks are {\em homogeneous networks} where the autonomous planar system at each vertex is {\em the same}. Winfree studied stability of a synchronised state whereas Turing derived the conditions for instability. Clearly the same questions can be asked in reverse, and indeed some researchers have already attempted to do so. However, all these examples rely on the autonomous dynamical units admitting (constant or periodic) attractors. We can therefore expect very different behaviour when the underlying planar systems are Hamiltonian.                 

\subsection{Model}
In this paper we will study networks with linear diffusive coupling where the underlying system at each vertex is described by a planar Hamiltonian. More specifically, we are interested in networks of the form
\begin{equation} \label{master}
\dot{x_i} = \frac{\partial h_i(x_i,y_i)}{\partial y_i} + D_x \sum_{j=1}^n W_{ij} (x_j-x_i), \quad \dot{y_i} = -\frac{\partial h_i(x_i,y_i)}{\partial x_i} + D_y \sum_{j=1}^n W_{ij} (y_j-y_i) .
\end{equation}
This dynamical system is a discretised reaction-diffusion equation and, if the variables $x_i,y_i$ represent concentrations of a pair of reactants $x,y$ at vertex $i$, the positive constants $D_x,D_y \in \mathbb{R}_{\geq 0 }$ can be interpreted as diffusion coefficients of $x,y$, respectively. We have already assumed that the network's weighted graph, $\mathcal{G}$, is undirected and connected. The weighted Laplacian associated with $\mathcal{G}$ is defined as the matrix having components $\Delta_{ij} = W_{ij}$ ($i \neq j$) and $\Delta_{ii} =- \sum_{j=1}^n W_{ij}$. The assumptions on $\mathcal{G}$ imply that the weighted Laplacian is negative-semidefinite with an orthonormal basis of eigenvectors $\{\xi^j\}$ corresponding to eigenvalues $\{\lambda_j\}$ and a kernel spanned by the normalised eigenvector $\xi_1 = (1/\sqrt{n},1/\sqrt{n},...,1/\sqrt{n})^T$. The subscript on $h_i$ indicates that in general the planar Hamiltonian system defining the reaction at each vertex of $\mathcal{G}$ need not be the same
\begin{equation} \label{planar}
\dot{x} = \frac{\partial h_i(x,y)}{\partial y}, \quad \dot{y} = -\frac{\partial h_i(x,y)}{\partial x} ,
\end{equation}
and so the general network is heterogeneous. When restricting to homogeneous networks we take $h_i(x,y) = h(x,y)$ for all $i \in V$, in which case (\ref{master}) corresponds to $n$ copies of the same underlying planar system (\ref{ham}). When $D_x=D_y = 0$ the total system (\ref{master}) uncouples and is conservative with Hamiltonian
\begin{equation} \label{totalham} H(x_1,y_1,...,x_n,y_n) = \sum_{i=1}^n h_i(x_i,y_i) . \end{equation} 
Uncoupled networks are integrable by the simple fact that each planar system is individually integrable and so the total system admits $n$ invariants of motion, $h_i(x_i,y_i)$. However, when $D_x, D_y > 0$ the coupled system (\ref{master}) is no longer Hamiltonian by construction.  

Networks related to (\ref{master}) have appeared in the literature several times before. We describe two important examples that are, to the best of our knowledge, the closest analogues of (\ref{master}). First is the class of systems introduced by Smereka \cite{Sme98} in the search for a Hamiltonian version of the Kuramoto model \cite{Kur84}. Together with those considered by Zanette-Hampton-Mikhailov \cite{Zan97,Ham99}, De Nigris-Leoncini \cite{Nig13,Nig15}, and Virkar-Restrepo-Meiss \cite{Vir15}, these consist of planar Hamiltonian systems coupled in such a way that the total network remains Hamiltonian. Our networks would fall into this class if we were to allow negative values for diffusion coefficients and take $D_x = -D_y$, because then (\ref{master}) can be written as Hamiltonian system of dimension $2n$ with Hamiltonian
\begin{equation} H(x_1,y_1,...,x_n,y_n) = \sum_{i=1}^n h_i(x_i,y_i) + D_x \sum_{i,j=1}^n W_{ij} (x_i-x_j)(y_i - y_j) . \end{equation}
The key difference between our networks and theirs is that \emph{complete synchronisation} of the latter would violate Liouville's Theorem, an important theorem in symplectic geometry concerning the preservation of phase space volume under the flow of a Hamiltonian vector field. Consequently, the authors of \cite{Ham99} introduced the concept of \emph{measure synchronisation} for the special case that the total network remains Hamiltonian. Conversely, in the next section we will demonstrate complete synchronisation for different choices of $h_i$ in (\ref{master}), which generalise the networks of coupled harmonic oscillators studied by Ren \cite{Ren08}. Ren's network is a homogeneous version of (\ref{master}) where one has $D_x=0$ and $h(x,y) = (y^2 + \omega^2 x^2)/2$ for some constant $\omega$. In these cases complete synchronisation is allowed because positive diffusion coefficients generate a flow that exponentially contracts the volume of phase space. 

Starting with steady-state solutions, because of $\lambda_1=0$ it will never be possible to establish asymptotic stability of a homogeneous equilibrium $(x_i,y_i) = (x^*,y^*)$ for all $i \in V$. For this solution not to be unstable it must be a centre of the underlying planar system, but this necessarily implies the Jacobian of the linearisation of (\ref{master}) must also have a pair of purely imaginary eigenvalues. Thus, the Hartman-Grobman Theorem also fails to apply in this case and the best one can hope to do is establish that instability emerges in the presence of diffusion. This is very similar to the famework of pattern formation in Turing's activator-inhibitor networks \cite{Tur52,Nak10}. For oscillations on the other hand, we have already described a change of coordinates $ (x_i,y_i) \to (I_i, \theta_i)$ such that $h_i(x,y) = h_i(I)$. Performing this transformation on the networks (\ref{master}) yields
\begin{equation} \label{actionangle}
\dot{I}_i = \sum_{j=1}^n \Delta_{ij} \left[ D_x \frac{\partial I_i}{\partial x_i} x_j + D_y \frac{\partial I_i}{\partial y_i} y_j  \right],  \quad \dot{\theta}_i = \omega_i(I_i) + \sum_{j=1}^n \Delta_{ij} \left[ D_x \frac{\partial \theta_i}{\partial x_i} x_j + D_y \frac{\partial \theta_i}{\partial y_i} y_j  \right]
\end{equation}
and since the $\dot{I}_i$ are not identically zero we \emph{do not} have $I_i = \mbox{const}$. Instead we can formally think of $I_i(t)$ parameterising a family of periodic orbits underlying the planar system (\ref{planar}) and at any particular time $t$ say that the system at vertex $i$ is in phase $\theta_i(t)$ of orbit $I_i(t)$. The doubling of phase space to accommodate the additional variables $I_i$ is summarised as follows: in networks of limit cycle oscillators each vertex $i$ arrives at the same periodic orbit (the limit cycle) in the absence of coupling, and only an adjustment of phase $\theta_i$ is required for synchronisation. In networks of planar Hamiltonian systems each vertex $i$ remains on a different periodic orbit in the absence of coupling and therefore both phase $\theta_i$ and periodic orbit $I_i$ must be adjusted for synchronisation. Thus, whilst networks of limit cycle oscillators can be reduced to simple models of coupled phase oscillators, networks of Hamiltonian oscillators can not. Moreover, if the system (\ref{master}) does synchronise it is not immediately clear which periodic orbit will result.

\section{Synchronisation}
In the previous section we discussed synchronisation without recourse to a concrete definition. Throughout the remainder of this paper will reserve the term synchronisation for what is commonly referred to as complete synchronisation, namely a solution 
\[ (x_1,x_2,...,x_n,y_1,y_2, ...,y_n) : \mathbb{R}_{\geq 0} \to \mathbb{R}^{2n} \]
to (\ref{master}) achieves {\em synchronisation} if all pairs $(x_i(t),y_i(t))$ become identical as $t \to \infty$. One should be aware that there are related notions of frequency synchronisation and general synchronisation that we shall not describe here. Synchronisation in the context defined above is usually only possible for homogenous networks of diffusively coupled oscillator of the form
\begin{equation} \dot{z}_i = F(z_i) + \sum_{j=1}^n W_{ij} U(z_i - z_j) . \end{equation}
Stability of a synchronous state $z_i(t) = s(t)$ for all $i \in V$, where $s(t)$ is a solution of the underlying system $\dot{s} = F(s)$, can be formulated in terms of the {\em master stability function} \cite{Pec98}. 

For networks of limit cycle oscillators we have already explained that $s(t)$ will always be an isolated limit cycle of $\dot{z}=F(z)$. When this system is planar Hamiltonian of the form (\ref{ham}) however, there will necessarily exist a family of periodic orbits and amongst these it is not at all clear to which oscillatory solution $s(t)$ corresponds. Families of periodic orbits are parameterised by constant values of $h(x,y)$ and will in general have a very complicated structure within different regions of the plane. To attack these problems, in this section we shall restrict ourselves to Hamiltonians of the form $h(x,y) = y^2/2 + P_k(x)$ where $P_k$ is a polynomial of degree $k$. The level curves of $h$ are rational for $k=1,2$, elliptic for $k=3,4$ and hyperelliptic for $k\geq 5$. We assume $k >1$ since the level curves have no closed components if $k=1$. In this case the quadratic theory $k=2$ produces a linear network for which we can solve and deduce the properties of synchronisation exactly. Subsequently, we shall present some numerical calculations for nonlinear networks, considering cases $k=3,4$ in particular, together with a few brief remarks about the general nonlinear theory.       
       
\subsection{Quadratic curves and linear theory} \label{sslinear}
In the case $k=2$ periodic orbits corresponding to level sets of the Hamiltonian
\begin{equation} \label{quadratic} h(x,y) = \frac{1}{2}y^2 + \frac{\omega^2}{2}x^2  \quad \omega \in \mathbb{R} \end{equation}
form a set of concentric ellipses surrounding an isolated centre at $h(x,y) =0$. The resulting network (\ref{master}) is linear and can therefore be solved exactly. However, we prefer to tease out more details to help us understand the nonlinear case. When calculating eigenvalues we consider a general quadratic Hamiltonian rather than (\ref{quadratic}) to uncover a set of relations that play an important role when we turn to pattern formation in the next section.

We begin by establishing a general synchronisation result for linear planar systems. In particular, we consider linear systems that in vector form with $\mathbf{x} = (x_1,x_2,...,x_n)^T$, $\mathbf{y} = (y_1,y_2,...,y_n)^T$ are given by
\begin{equation} \label{linear} \left( \begin{array}{c} \dot{\mathbf{x}}  \\ \dot{\mathbf{y}}  \end{array} \right) = \left( \begin{array}{cc} a\cdot \mathbf{I}_n+D_x \Delta  & b\cdot \mathbf{I}_n \\ c\cdot \mathbf{I}_n & d\cdot \mathbf{I}_n + D_y \Delta  \end{array} \right)  \left( \begin{array}{c} \mathbf{x}  \\ \mathbf{y} \end{array} \right)   \quad a,b,c,d \in \mathbb{R} \end{equation}
where $ \mathbf{I}_n$ is the $n \times n$ identity matrix. Assuming this linear network synchronises we ask: to which solution of the underlying planar system
\begin{equation} \label{short}
\dot{x} = ax+by, \quad \dot{y} = cx+dy 
\end{equation} 
does the system (\ref{linear}) synchronise? The answer to this question is reassuringly pleasant.
\begin{claim} \label{average}
Let $\{x_j(0),y_j(0)\}_{j=1,2,...,n}$ be a set of initial conditions for the linear network (\ref{linear}). If the solution to (\ref{linear}) with these initial conditions synchronises, then it synchronises to the solution of the planar system (\ref{short}) with initial conditions $x(0) = \frac{1}{n} \sum_{j=1}^n x_j(0)$, $y(0) = \frac{1}{n} \sum_{j=1}^n y_j(0)$. 
\end{claim} 
\begin{proof}
We decompose solutions of (\ref{linear}) using orthonormal eigenvectors $\xi^j$ of the weighted Laplacian matrix $\Delta$ corresponding to eigenvalues $\lambda_j$ (chosen so that $\Delta \xi_j = - 2\lambda_j \xi_j$),
\[ \mathbf{x}  = \sum_{j=1}^n u_j(t) \xi^j, \quad \mathbf{y}  = \sum_{j=1}^n v_j(t) \xi^j, \]
and obtain $n$ uncoupled pairs of planar systems
\[ \dot{u}_j = (a-2\lambda_j D_x) u_j + b v_j, \quad \dot{v}_j = cu_j +(d-2 \lambda_j D_y) v_j .\]
Recall that the eigenvector corresponding to the simple eigenvalue $\lambda_1=0$ is $\xi^1= (\frac{1}{\sqrt{n}}, ..., \frac{1}{\sqrt{n}})^T$ and so the system approaches a synchronous state as $t\to \infty$ if and only if $u_j(t),v_j(t) \to 0$ for all $j>1$. Thus $x_j(t) \to u_1(t)/\sqrt{n}$ and $y_j(t) \to v_1(t)/\sqrt{n}$ for all $j$, where $(u_1(t),v_1(t))$ is a solution of (\ref{short}). To find the initial conditions that determine this solution, simply observe that by definition of the decomposition we have
\[ \frac{1}{\sqrt{n}}u_1(0) = \frac{1}{\sqrt{n}} \langle \mathbf{x}(0) ,\xi_1  \rangle = \frac{1}{n} \sum_{j=1}^n x_j(0) , \quad \frac{1}{\sqrt{n}}v_1(0) = \frac{1}{\sqrt{n}} \langle \mathbf{y}(0) ,\xi_1  \rangle = \frac{1}{n} \sum_{j=1}^n y_j(0) . \]
\qed
\end{proof}

We next turn to the homogeneous version of the network (\ref{master}) equipped with a general quadratic Hamiltonian
\[  h(x,y)=\frac{1}{2}(Ax^2+2Bxy+Cy^2), \]
which becomes of the form (\ref{linear}) if $a=B$, $b = C$, $c=-A$, and $d=-B$. In this case the underlying planar system (\ref{planar}) is well known to admit periodic solutions when $AC-B^2>0$, and we therefore assume this inequality throughout. Decomposing solutions following the proof of Claim \ref{average} again results in a system of uncoupled equations whose eigenvalues come in pairs
\[ \Lambda^\pm_j = -(D_x+D_y)\lambda_j \pm \sqrt{(D_x+D_y)^2 \lambda_j^2-k(\lambda_j)} ,\]
where $k(\lambda_j)=4  D_xD_y \lambda_j^2+ 2 B (D_x-D_y)\lambda_j  + AC - B^2$ and corresponding eigenvectors are
\[ \mathbf{k}^\pm_j =  \left( \begin{array}{c} 1  \\ (\Lambda^\pm_j - B + 2 \lambda_j D_x)/C \end{array} \right) . \]
There is a purely imaginary pair $\Lambda^{\pm}_1$ corresponding to $\lambda_1=0$ given by
\[ \Lambda^{\pm}_1 = \pm i \sqrt{AC-B^2} , \]
but for the remaining $\lambda_2, ..., \lambda_n>0$ the pairs $\Lambda^{\pm}_j$ can be separated into five classes depending on the value of $k(\lambda_j)$:

1) case $k(\lambda_j)<0$: $\Lambda^{\pm}_j$ are real with $\Lambda^+_j>0$, $\Lambda^-_j <0$;

2) case $k(\lambda_j)=0$: $\Lambda^{\pm}_j$ are real with $\Lambda^+_j=0$, $\Lambda^-_j <0$;

3) case $0<k(\lambda_j)< (D_x+D_y)^2\lambda_j^2$: $\Lambda^{\pm}_j$ are real with $\Lambda^+_j<0$, $\Lambda^-_j <0$;

4) case $k(\lambda_j)=(D_x+D_y)^2\lambda^2_j$: $\Lambda^{\pm}_j$ are real with $\Lambda^+_j = \Lambda^-_j <0$;

5) case $(D_x+D_y)^2\lambda^2_j < k(\lambda_j)$: $\Lambda^{\pm}_j$ are complex conjugates with $Re(\Lambda^{\pm}_j) < 0$.
\\
The real parts of these eigenvalue pairs govern the large time behaviour of the solutions $u_j(t),v_j(t)$ and if $Re(\Lambda^{\pm}_j)<0$ then $u_j(t),v_j(t) \to 0 $ as $t\to \infty$. The following is therefore a modest generalisation of Ren's result \cite{Ren08} for coupled harmonic oscillators and answers the question posed at the beginning of this subsection.
\begin{claim}
Let $\{x_j(0),y_j(0)\}_{j=1,2,...,n}$ be a set of initial conditions for the homogeneous network (\ref{master}) with quadratic hamiltonian (\ref{quadratic}). Then, as $ t \to \infty$, we have
\[ x_j(t) \to \bar{x}(0) \cos \omega t +\frac{\bar{y}(0)}{\omega}  \sin \omega t , \quad  y_j(t) \to \bar{y}(0) \cos \omega t - \bar{x}(0) \omega \sin \omega t , \]
for all $j \in V$ where
\[ \bar{x}(0) = \frac{1}{n} \sum_{i=1}^n x_i(0) , \quad \bar{y}(0) = \frac{1}{n} \sum_{i=1}^n y_i(0) . \]

\end{claim}
\begin{proof}
When $B=0$ we have $k(\lambda_j)>0$ for $j>0$ and therefore eigenvalue pairs $\Lambda^{\pm}_j$ fall into classes 3) - 5). Consequently $Re(\Lambda^{\pm}_j)<0$, and using the argument presented in the proof of Claim \ref{average} this implies the system synchronises. By the same argument it must synchronise to the solution of the underlying linear system with average initial conditions. 
\qed
\end{proof}

Let us summarise the results of this subsection. We first verified that synchronisation of the network (\ref{linear}) necessarily implies synchronisation to the solution of the linear planar system (\ref{short}) whose initial conditions are given by averaging initial conditions across the network. We then calculated the eigenvalues associated with the choice quadratic Hamiltonian (\ref{quadratic}) to deduce that in this case the network (\ref{master}) always synchronises to the ``average'' periodic orbit of the underlying planar Hamiltonian system (\ref{ham}). For quadratic hamiltonians this answers the problem raised at the end of section 2 and explains how each action variable, $I_i$, must be adjusted during synchronisation. At first glance it is tempting to speculate that the same ``averaging'' result extends to all $k>2$. Using a simple argument we will show this not to be the case however; determining the synchronised state in nonlinear networks (\ref{master}) appears to remain an extremely complicated problem that in general may only be solved numerically.

\subsection{Elliptic curves and nonlinear theory} \label{ssnonlinear}

General nonlinear networks (\ref{master}) can also be written in vector form
\[  \left( \begin{array}{c} \dot{\mathbf{x}}  \\ \dot{\mathbf{y}}  \end{array} \right) = \left( \begin{array}{cc} B\cdot \mathbf{I}_n+D_x \Delta  & C\cdot \mathbf{I}_n \\ -A\cdot \mathbf{I}_n & -B\cdot \mathbf{I}_n + D_y \Delta  \end{array} \right)  \left( \begin{array}{c} \mathbf{x}  \\ \mathbf{y} \end{array} \right)  + \left( \begin{array}{c} \frac{\partial \widetilde{H} }{\partial \mathbf{y}}  \\ -\frac{\partial \widetilde{H}}{\partial \mathbf{x}} \end{array} \right)  \]
with
\[ \widetilde{H} = \sum_{i=1}^n \left[ h(x_i,y_i) - \frac{1}{2}(Ax^2_i + 2B x_iy_i + Cy^2_i)  \right] \]
and $A,B,C \in \mathbb{R}$ chosen so that $\widetilde{H}$ contains no quadratic terms. Making the same decomposition as in the proof of Claim \ref{average} transforms the system
\begin{equation} \label{nonlinear} \dot{u}_j = (B+\lambda_j D_x) u_j + C v_j + \frac{\partial \widetilde{H}}{\partial v_j}, \quad \dot{v}_j = -Au_j +(-B+ \lambda_j D_y) v_j - \frac{\partial \widetilde{H}}{\partial u_j},\end{equation}
where $\widetilde{H}=\widetilde{H}(u_1,...,u_n,v_1,...,v_n)$ is now considered as a function of the $u_j,v_j$. Explicitly, denoting the $i$th component of the eigenvector $\xi^j$ by $\xi^j_i$, we have
\[ \widetilde{H}  =  \sum_{i=1}^n \left[ h\left(\sum_{j=1}^n u_j \xi^j_i,\sum_{j=1}^n v_j \xi^j_i \right) - \frac{A}{2}\left(\sum_{j=1}^n u_j \xi^j_i \right)^2 - B \left(\sum_{j=1}^n u_j \xi^j_i \right)\left(\sum_{j=1}^n v_j \xi^j_i \right) - \frac{C}{2}\left(\sum_{j=1}^n v_j \xi^j_i \right)^2  \right] \]
indicating that coupling has been shifted to the nonlinear terms of (\ref{nonlinear}) in this representation of the network. By the same argument used in the proof of Claim \ref{average}, if the system is to synchronise to a periodic orbit $(u(t),v(t))$ of the underlying planar Hamiltonian system (\ref{ham}), we must have $u_j(t),v_j(t) \to 0$ as $t \to \infty$ for all $j>1$. Unlike the linear case however, the pair of equations corresponding to the zero mode $\lambda_1=0$ do not decouple from the rest
\[ \dot{u}_1 = Bu_1 + C v_1 + \frac{\partial \widetilde{H}}{\partial v_1}, \quad \dot{v}_1 = -Au_1 -Bv_1 - \frac{\partial \widetilde{H}}{\partial u_1}, \]  
and so are not of the form (\ref{ham}) unless $t=\infty$. Since the asymptotic behaviour of the trajectory $(u_1(t),v_1(t))$ determines the periodic orbit $(u(t),v(t))$, identifying its dependence on initial conditions amounts to solving the entire network (\ref{nonlinear}). This means in general there will be no obvious way to identify the synchronised orbit of the network because each system will have its own functional value of $\widetilde{H}$. In the remainder of this subsection we therefore focus on several examples of elliptic curves and use a numerical approach to determine some properties of synchronisation in these special cases. For simplicity we only study pairs of coupled oscillators ($n=2$).          

{\em Case $k=3$}. If we suppose that the level sets of $h$ contain a continuous family of closed curves then the planar Hamiltonian system (\ref{ham}) has two critical points, a centre and a saddle, which may be chosen (without loss of generality) at $(-1,0)$ and $(1,0)$, respectively. This implies $P_3(x) = -x^3/3 + x$ so that $h$ is of the form
\begin{equation} \label{elliptic} h(x,y) = \frac{1}{2} y^2 -\frac{1}{3} x^3 + x   \end{equation}
and the network (\ref{master}) becomes 
\begin{equation} \label{ellipticnetwork} \dot{x}_i = y_i + D_x \sum_{j=1}^n W_{ij} (x_j- x_i) , \quad \dot{y}_i = x^2_i - 1 + D_y \sum_{j=1}^n W_{ij} (y_j - y_i) . \end{equation}
Numerical simulations of this network for a pair of coupled planar systems are presented in Figure \ref{figure1} and demonstrate that, whilst the network synchronises, the resulting periodic trajectory can not be obtained by averaging initial conditions $(x_1(0),y_1(0))$ and $(x_2(0),y_2(0))$. The time-dependent synchronisation coefficient (absolute distance between a specified pair of trajectories) for the pair $(x_1(t),y_1(t))$ and $(x_2(t),y_2(t))$ approaches zero in exponential time reflecting synchronisation of the network, but oscillates wildly when calculated for either $(x_i(t),y_i(t))$ and the average periodic orbit. This is in contrast with the quadratic case where all the possible combinations of time-dependent synchronisation coefficients decay to zero in exponential time. On the basis of extensive numerical simulations with a large variety of initial conditions, parameters values, and large $n>2$ (data not shown) we conjecture that synchronisation always occurs for generic elliptic curves $h(x,y) =y^2/2 + P_3(x)$ where there exists a single family of periodic orbits, provided that $D_x \approx D_y$. As expected however, the synchronised trajectory is not obtained by averaging initial conditions and so already for $k=3$ we see a departure from the quadratic case $k=2$.     

\begin{figure}
\centering
  \includegraphics[width=\textwidth]{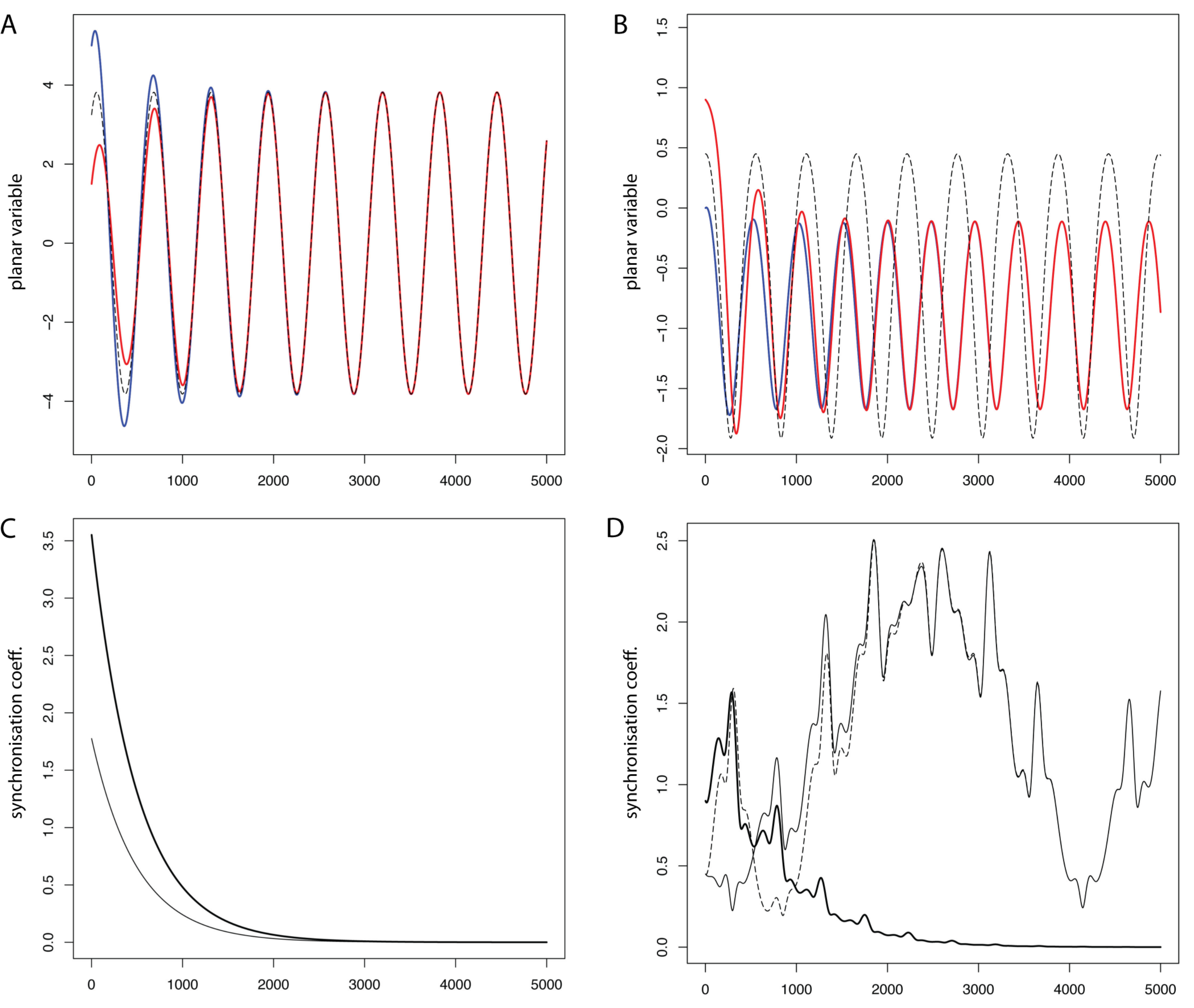}
  \caption{ \label{figure1} Examples of synchronisation. A) Quadratic Hamiltonian (\ref{quadratic}) with $\omega = 1$ results in synchronisation of trajectories $x_1(t)$ (blue) and $x_2(t)$ (red), which approach the periodic trajectory of the underlying planar Hamiltonian system obtained by averaging initial conditions $(x_1(0),y_1(0)) = (5.0,2.3)$ and  $(x_2(0),y_2(0)) = (1.5,1.7)$ (black dashed line). B) Elliptic Hamiltonian (\ref{elliptic}) results in synchronisation of trajectories $x_1(t)$ (blue) and $x_2(t)$ (red), but these remain far from the periodic trajectory obtained by averaging initial conditions $(x_1(0),y_1(0)) = (0.0,0.0)$ and  $(x_2(0),y_2(0)) = (0.9,0.0)$ (black dashed line). Synchronous behaviour for C) quadratic and D) elliptic Hamiltonians can be visualised by calculating time-dependent synchronisation coefficients for each pair: $(x_1,y_1)$ and $(x_2,y_2)$, bold; $(x_1,y_1)$ and the average periodic orbit, regular line; $(x_2,y_2)$ and the average periodic orbit, dashed line. All equations were integrated using the Runge-Kutta method with a time step of $0.01$ and diffusion coefficients were set to $D_x=D_y =0.1$.}
\end{figure}

In the quadratic case one can confirm that, for any pair of initial conditions $(x_i(0),y_i(0))$ ($i=1,2$), the total Hamiltonian (\ref{totalham}) is always monotonically decreasing along trajectories. This turns out not to be true for elliptic curves where it is possible to find initial conditions such that the total Hamiltonian 
\[ H(x_1,y_1,x_2,y_2) = h(x_1,y_1) + h(x_2,y_2) \]
$increases$ in time, i.e., $H(t+t')>H(t)$ for some $t,t'>0$. Indeed, the total derivative of $H$ along the trajectories of (\ref{ellipticnetwork}) with $n=2$ is given by 
\[ \frac{dH}{dt} = -D_y(y_1 -y_2)^2 +D_x(x_1-x_2)(x_1^2 - x_2^2)\]
and may be positive provided initial conditions are chosen appropriately. Let, for example, $(x_1(0),y_1(0))=(0,0)$ lie on the periodic orbit passing through the origin and $(x_2(0),y_2(0))$ lie on a periodic orbit (i.e., $-2/3 < h(x_2(0),y_2(0)) < 2/3$) with $0<x_2(0) < 1$ and $D_y y_2^2(0)> D_x x_2^3(0)$. One can check that these inequalities are mutually compatible and this choice of initial conditions yields $dH/dt(0)>0$. Moreover, provided $D_x \approx D_y$ are sufficiently large, trajectories will not have time to leave the region with $dH/dt > 0$ before synchronisation occurs. This observation (which has been verified numerically) implies an overall increase in the value of $H$ upon synchronisation, indicating that the total Hamiltonian (\ref{totalham}) is not always decreasing as one might originally suspect. Similar exceptions to this naive assumption, which holds for the quadratic case, can be found in cases $k>3$.   

{\em Case $k=4$}. In this case the generic $h$ is of the form
\[ h(x,y) = \frac{1}{2} y^2 + \frac{a}{4}x^4 + \frac{b}{3}x^3 + \frac{c}{2} x^2 \]
with $a \neq 0$. There are five types of continuous families of closed curves contained within the level sets of $h$ depending on the parameters $(a,b,c)$ (see pg. 106 in \cite{Chr00}). To study synchronisation in an example where there can exist multiple families of periodic orbits we choose to numerically investigate networks with the planar Hamiltonian
\[ h(x,y) = \frac{1}{2}y^2 + \frac{1}{4}x^4 - \frac{1}{2} x^2 ,\]
whose level sets correspond to closed curves surrounding a separatrix (Figure \ref{figure2}A). There are three critical points of this Hamiltonian, two centres at $(-1,0)$ and $(1,0)$, and a saddle at the origin $(0,0)$. The separatrix given by $h(x,y)=0$ passes through the saddle and separates three families of periodic orbits: one contained in each of the two lobes, consisting of concentric ovals surrounding each centre, and a family of larger curves that emanate outwards from the separatrix. Numerical simulations suggest that, whenever initial conditions are such that all vertices of the network begin on periodic solutions of the {\em same} family, synchronisation always occurs in the presence of diffusion $D_x \approx D_y$ just as it did for $k=3$ (data not shown). Once again the resulting synchronised orbit can not be obtained by averaging initial conditions. On the other hand, a quite peculiar form of synchronisation arises when initial conditions determine periodic orbits in {\em different} families. As indicated in Figure \ref{figure2}, when a pair of planar systems are initialised on periodic orbits lying inside different lobes of the separatrix (Figure \ref{figure2}A) they will not synchronise for small values of the diffusion coefficients (Figure \ref{figure2}B). However, once the strength of diffusion surpasses a certain threshold the network undergoes a phase transition similar to that of the Kuramoto model \cite{Kur84} and synchronises to a periodic orbit in one of the respective families (Figure \ref{figure2}C). Even more surprising is the fact that, as the value of the diffusion coefficient increases further still, there is a second phase transition for the same initial conditions where the synchronised periodic ``jumps'' across the saddle and into the other lobe (Figure \ref{figure2}D)! A related phenomenon occurs if initial conditions lie on either side of the separatrix (i.e., one inside a lobe and the other on a larger amplitude orbit outside, Figures \ref{figure2}E and \ref{figure2}F). When the strength of diffusion is increased the trajectory beginning on the large orbit is pulled inside the separatrix and on to a member of the family inside one of the lobes (Figures \ref{figure2}G and \ref{figure2}H). Repeated simulations with randomised initial conditions suggest that the size of the diffusion coefficient relative to the distance between initial conditions appears to determine the time and location of the crossing, and ultimately whether or not the network synchronises. This again is reminiscent of the heterogenous Kuramoto model where distributions amongst families of periodic orbits are replaced by distributions of intrinsic frequencies.

\begin{figure}
\centering
  \includegraphics[width=0.65\textwidth]{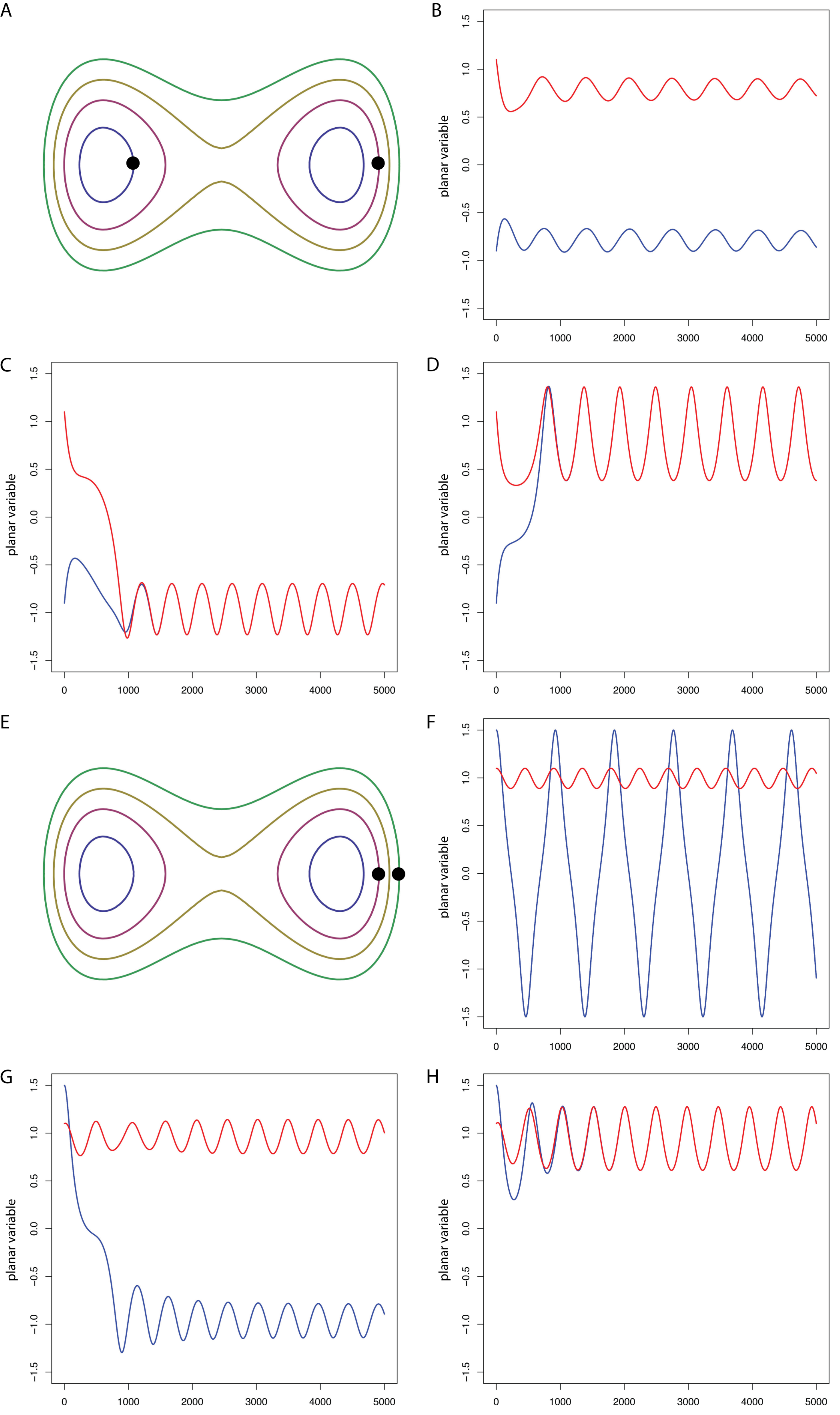}
  \caption{ \label{figure2} Synchronisation of trajectories initialised on different families of periodic orbits separated by a separatrix. A) Initial conditions $(x_1(0),y_1(0)) = (-0.9,0.0)$ and  $(x_2(0),y_2(0)) = (1.1,0.0)$ for first three panels where synchronisation behaviour depends on the strength of diffusion: B) $D_x=D_y=0.3$, C) $D_x=D_y=0.4$, and D) $D_x=D_y=0.5$. E) Initial conditions $(x_1(0),y_1(0)) = (1.5,0.0)$ and  $(x_2(0),y_2(0)) = (1.1,0.0)$ for last three panels where synchronisation behaviour depends on the strength of diffusion: F) $D_x=D_y=0.0$, G) $D_x=D_y=0.1$, and H) $D_x=D_y=0.2$. The planar variable $x_1(t)$ is always in blue with $x_2(t)$ always in red . Equations were integrated using the Runge-Kutta method with a time step of $0.01$.}
\end{figure}

Another peculiar type of behaviour present even in this simple choice of model with $n=2$ is the phenomenon of {\em oscillation quenching} \cite{Kos13}. It seems that when initial conditions are symmetric about the saddle, oscillations that persist in the absence of diffusion (Figure \ref{figure3}A) become damped when diffusion increases and trajectories are pulled toward different centres inside the separatrix (Figure \ref{figure3}B). As diffusion coefficients increase further this desynchronised state of the network makes a transition toward a synchronised steady state where both vertices occupy the saddle (Figures \ref{figure3}C and \ref{figure3}D). This form of oscillation quenching, which almost certainly depends on the initial conditions $(x_1(0),y_1(0))$ and $(x_2(0),y_2(0))$ being exact mirror images about the saddle, is a scenario where a nonlinear network does in fact synchronise to the ``average orbit'', albeit one that is constant rather than time-periodic. It is example of {\em amplitude death} \cite{Kos13}. That we already observe such a variety of phenomena for the case $k=4$, $n=2$ implies that with a generic choice of hamiltonian the nonlinear network (\ref{master}) may exhibit a spectrum of exotic behaviours. A systematic investigation of these phenomena by numerical means is not possible however, as we have seen that different behaviours invariably depend on the types of Hamiltonian, initial conditions and diffusion coefficients. It might be feasible to deduce synchronisation or oscillation quenching conditions for a particular class of Hamiltonians without recourse to numerical integration, but in general the absence of attractors in the underlying planar system makes it incredibly hard to predict what the resulting solution will be. A likened challenge would be to determine synchronous solutions of coupled chaotic systems \cite{Boc02}.           

\begin{figure}
\centering
  \includegraphics[width=\textwidth]{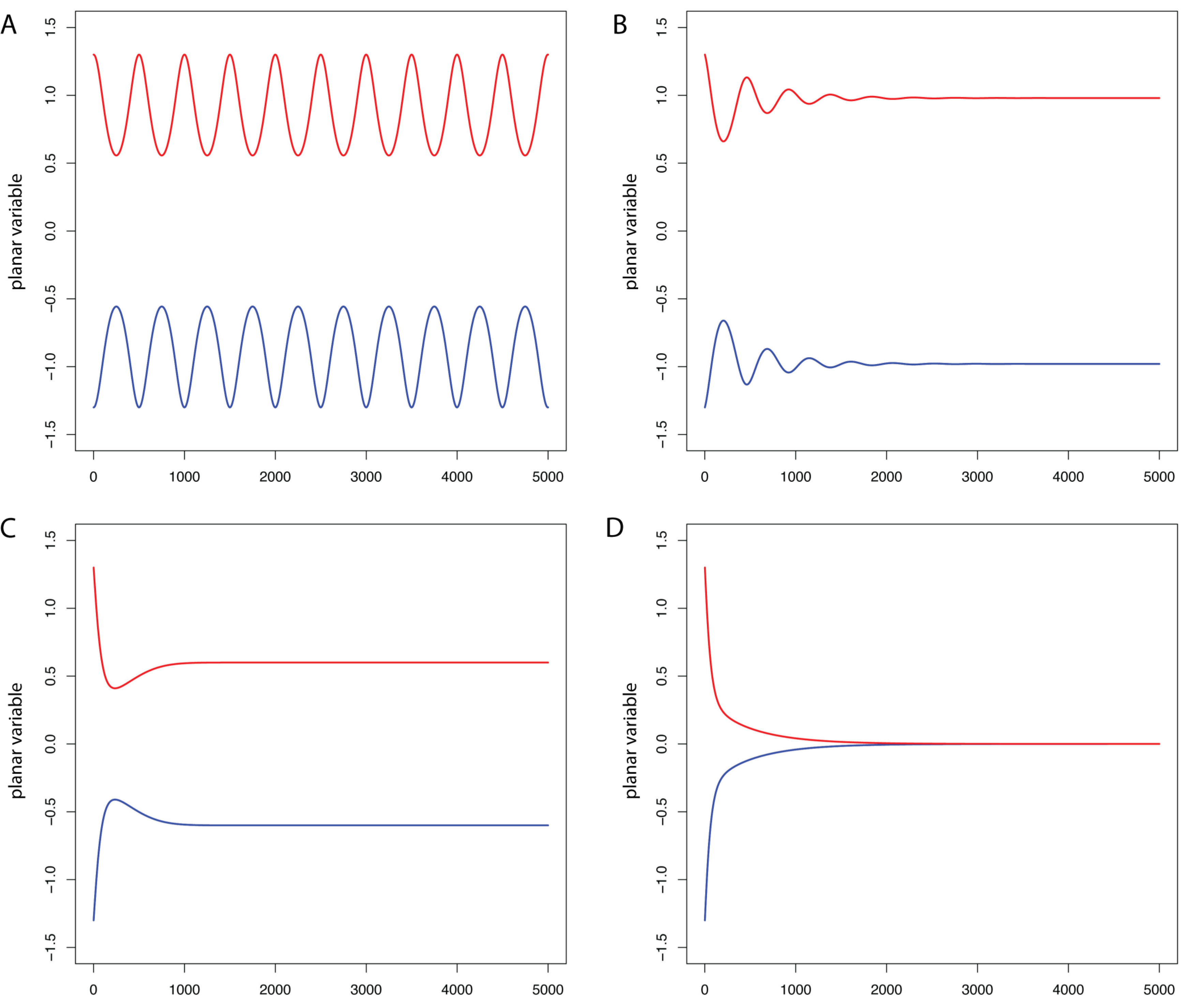}
  \caption{ \label{figure3} Examples of oscillation death. Initial conditions $(x_1(0),y_1(0)) = (-1.3,0.0)$ and  $(x_2(0),y_2(0)) = (1.3,0.0)$ produce anti-phase oscillations in the absence of diffusion A) and various forms of oscillation death with increasing diffusion coefficients: B) $D_x=D_y=0.1$, C) $D_x=D_y=0.4$, and D) $D_x=D_y=0.6$. The planar variable $x_1(t)$ is always in blue with $x_2(t)$ always in red. Equations were integrated using the Runge-Kutta method with a time step of $0.01$.}
\end{figure}

\section{Pattern formation}

As in the previous section, when discussing pattern formation we only consider networks of the form (\ref{master}) that are homogeneous ($h_i = h$ for all $i \in V$). On top of this it will be useful to introduce terminology distinguishing a {\em homogeneous state} from a {\em heterogeneous state}, which are properties of a trajectory rather than the network structure. Quite simply, we say the network is in a homogeneous state if it is synchronised in the usual sense so that $(x_i,y_i) = (\bar{x},\bar{y})$ for all $i \in V$. It is in a heterogeneous state otherwise. These are standard terminologies in the pattern formation literature. Thus, whilst synchronisation describes the evolution of a heterogeneous state toward a homogeneous state, pattern formation is said to occur when a network perturbed from a homogeneous state evolves towards a stable heterogeneous state. Pattern formation has been studied ever since Turing's conception of the idea in 1952 \cite{Tur52} and is classically concerned with situations where the homogeneous state corresponds to an asymptotically stable equilibrium of the underlying planar system. Many variations have appeared since Turing's work (e.g. \cite{Mim78,Zha95,Mei00,Wyl07,Pet15}), but in all cases to date the planar system underlying the network has always admitted an attractor and is therefore quite different from (\ref{master}). In this section we shall introduce an analogous concept of Turing instability for networks where the underlying planar system is Hamiltonian. Although, like Turing, we begin by studying instabilities originating from stationary homogeneous states it will quickly become apparent that the more natural concept is homogeneous states that are allowed to vary in time.                        

\subsection{Stationary homogeneous states} \label{stationary}

It does not take anything more than a minor adaptation of Turing's work \cite{Tur52} to derive conditions sufficient for diffusion-driven instability of a non-hyperbolic equilibrium, and we do so here. In analogy with the standard procedure for activator-inhibitor networks (see Chapter 14.3 in \cite{Mur89}) we assume that the underlying planar system (\ref{ham}) admits at least one centre, $(x^*,y^*)$, which implies $\mbox{det}[\mbox{Hess}(h)(x^*,y^*)]>0$. Linearisation of the full system (\ref{master}) about the stationary homogeneous state $(x_i,y_i) = (x^*,y^*)$ for all $i \in V$ results in the same linear dynamics that were studied in subsection \ref{sslinear}  
\[ \left( \begin{array}{c}
\dot{\mathbf{\delta x}}  \\
\dot{\mathbf{\delta y}}  \end{array} \right) =\left( \begin{array}{cc} h_{xy}\cdot \mathbf{I}_n + D_x \Delta & h_{yy}\cdot \mathbf{I}_n  \\ -h_{xx}\cdot \mathbf{I}_n  & -h_{xy}\cdot \mathbf{I}_n  + D_y \Delta \end{array} \right)  \left( \begin{array}{c}
\mathbf{\delta x}  \\
\mathbf{\delta y}  \end{array} \right). \]
Here subscripts refer to mixed partial derivatives of the Hamiltonian evaluated at the centre: $h_{xx} = \partial^2h(x^*,y^*)/\partial x^2$, $h_{xy} = \partial^2h(x^*,y^*)/\partial x \partial y$,  and $h_{yy} = \partial^2h(x^*,y^*)/ \partial y^2$. We may therefore make identifications $A=h_{xx}$, $B=h_{xy}$, $C=h_{yy}$ and use the classification of eigenvalues provided in subsection \ref{sslinear} to evaluate stability of the stationary homogeneous state as a function of the diffusion coefficients. Given that classes 2) - 5) correspond to dynamics that can not be concretely classified as unstable, the only remaining option as a sufficient condition for instability is
\[ k(\lambda_j) = 4  D_xD_y \lambda_j^2+ 2 h_{xy} (D_x-D_y)\lambda_j  +  \mbox{det}[\mbox{Hess}(h)(x^*,y^*)]  <0 . \]
Finding the critical value of this parabola $k(\lambda_j)$ and imposing that it be less than zero yields
\begin{equation} \label{condition}  \frac{(D_x-D_y)^2(h_{xy})^2}{4D_xD_y} > \mbox{det}[\mbox{Hess}(h)(x^*,y^*)] , \end{equation}
which we recognise as the analogue of the Turing instability condition for networks of the form (\ref{master}). However we note that, since the Laplacian spectrum is discrete, the actual occurrence of the instability (appearance of a positive eigenvalue $\Lambda_j^+$) depends on the value of $\lambda_j$ that corresponds to $\Lambda_j^+$ crossing the horizontal axis. This is illustrated in Figure \ref{figure4}.

\begin{figure}
\centering
  \includegraphics[width=0.7\textwidth]{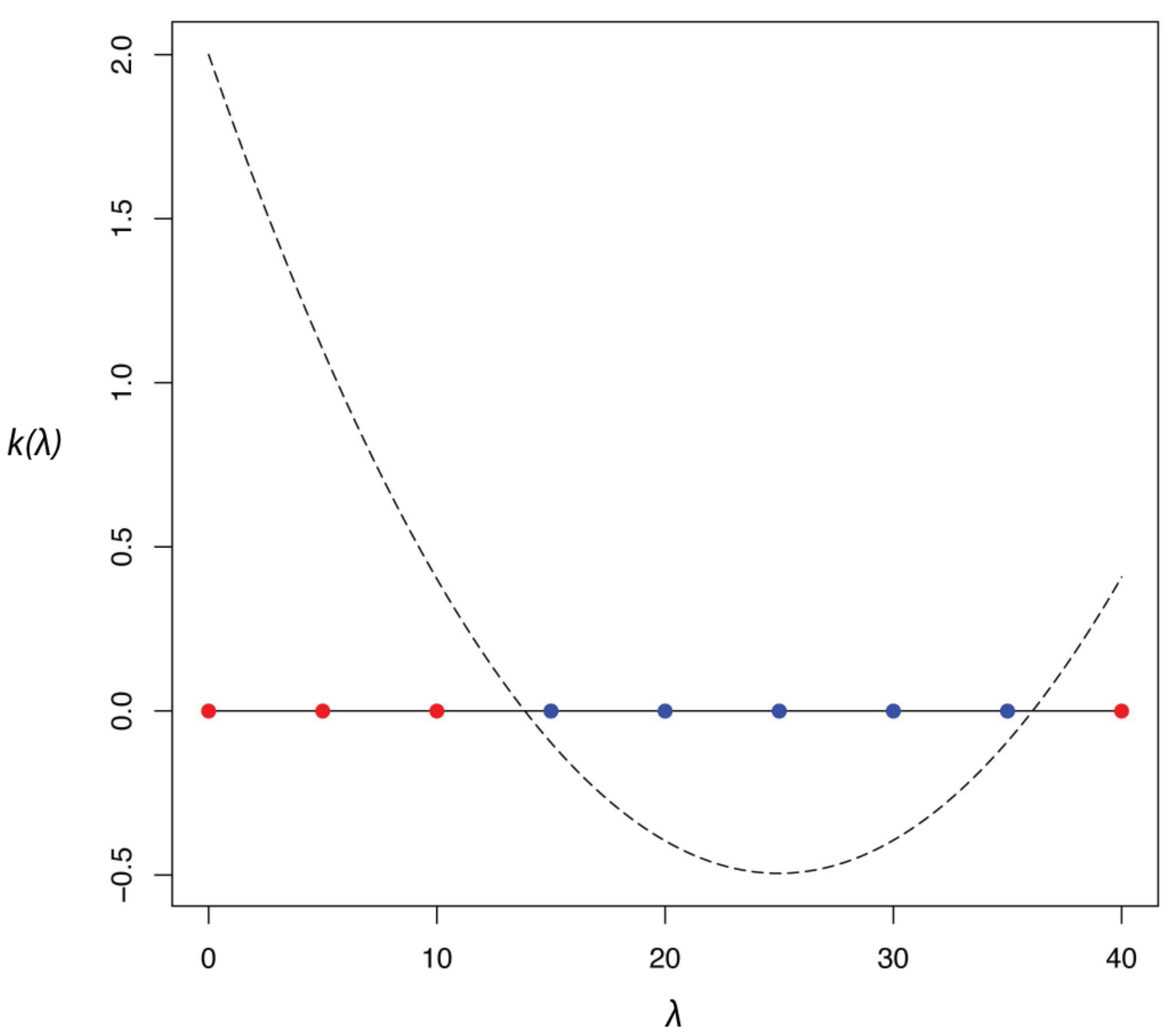}
  \caption{ \label{figure4} Example of a case where a subset of Laplacian eigenvalues $\lambda_j$ result in a $\Lambda_j^+$ whose real part crosses the horizontal axis to become positive. The parabola $k(\lambda)$ with $D_x = 0.001$, $D_y = 1.0$, $h_{xy} = 0.1$ and $\mbox{det}[\mbox{Hess}(h)(x^*,y^*)]=2.0$ is plotted as a dashed line and the solid line denotes the zero axis. In this example, eigenvalues are given by the formula $\lambda_j = 5(j-1)$ and values $\lambda_3=15,\lambda_4=20,\lambda_5=25,\lambda_6=30,\lambda_7=35$ that result in a negative value of $k(\lambda_j)$ and unstable modes are represented by blue dots. Eigenvalues $\lambda_0=0,\lambda_1=5,\lambda_2=10,\lambda_8=40$ resulting in $k(\lambda_j) \geq 0$ are coloured red.}
\end{figure}   

Of course, in the usual setting of activator-inhibitor networks \cite{Nak10} the centre $(x^*,y^*)$ is replaced by an asymptotically stable critical point of an underlying {\em dissipative} planar system so that the zero mode eigenvalues $\Lambda_1^\pm$ always have negative real part. This means that in the absence of diffusion the stationary homogeneous state remains stable under small perturbations, whilst in our case perturbations of the homogeneous centre can generate heterogeneous distributions of sustained oscillations. Homogeneity is preserved in both types of networks when diffusion coefficients are nonzero however, at least up to the point that the ratio $D_x/D_y$ crosses the threshold value for instability. The difference is that a homogeneous state of (\ref{master}) need not remain stationary since it is always susceptible to perturbations in the direction of the zero mode. An exception may arise when the perturbation corresponds to a symmetry operation under which the Hamiltonian and its critical point remain invariant, but this does not reflect the general case. Thus, we may still discuss pattern formation (i.e., breakdown of homogeneity as diffusion coefficients diverge) in networks of the type (\ref{master}), although it is more natural to do so from the point of view of time-dependent trajectories. We shall consider this problem in the next subsection. For now we simply remark that homogeneous states in networks of planar Hamiltonian systems remain homogeneous following perturbation provided that $D_x \approx D_y$, but may evolve towards a stable heterogeneous state when diffusion coefficients satisfy (\ref{condition}). Several authors have recently studied the effect of discrete network topology on pattern morphology \cite{Nak10,Asl14,Con16}.   
 
\subsection{Time-dependent homogeneous states} 
In the previous subsection we suggested that conventional Turing pattern formation (where patterns originate from a stationary homogeneous state) is not the natural object of study for networks of the form (\ref{master}) that do not admit asymptotically stable homogenous equilibria. Consequently, the relevant concept of pattern formation is one describing diffusion-driven instability of a {\em time-dependent} homogeneous state that we assume at any given time is a periodic orbit of (\ref{ham}). By analogy with Turing's work we want to deduce conditions for instability of this periodic homogeneous state. Closely related is the generation of Turing-type instabilities from a limit cycle \cite{Nak14,Cha15}. The authors of \cite{Cha15} argue that oscillation death, the type of oscillation quenching where an initially homogeneous state of diffusively coupled oscillators evolve toward a stationary heterogeneous state, is nothing more than a Turing instability for the first return map of the periodic homogeneous state. Thus, the concepts behind the master stability function may be tweaked slightly to rationalise a condition for instability of a periodic homogeneous state. By definition this is just the condition that diffusive coupling shifts the largest Floquet exponent above the horizontal axis (these are equivalent to the eigenvalues of the linearised first-return map). Floquet exponents are notoriously difficult to compute compared with eigenvalues $\Lambda_j^\pm$ however, and therefore approximated numerically for several different limit cycle oscillator networks in \cite{Cha15}. 

In contrast to the underlying planar system admitting an asymptotically stable limit cycle, here we have assumed a homogeneous state of the form $(x_i,y_i) = (\bar{x}(t),\bar{y}(t))$ for all $i \in V$ where $(\bar{x}(t),\bar{y}(t))$ is a periodic orbit of (\ref{ham}). Therefore all Floquet exponents necessarily lie on the horizontal axis when $D_x=D_y=0$. This mirrors the situation encountered in subsection \ref{stationary} where all eigenvalues $\Lambda_j^\pm$ were found to be purely imaginary in the absence of diffusion. Consequently, the homogeneous state can not be classified as stable unless we take nonzero diffusion coefficients. It is then reasonable to ask for conditions where diffusion results in the real part of at least one Floquet exponent becoming positive. The authors of \cite{Cha15} suggest using an analogue of condition (\ref{condition}) to determine when this occurs. For planar Hamiltonian systems the equivalent of this condition would be 
\begin{equation} \label{conditiontime}  \frac{(D_x-D_y)^2}{4D_xD_y} \left\langle \frac{\partial ^2 h(\bar{x},\bar{y})}{\partial x \partial y} \right\rangle ^2> \langle \mbox{det}[\mbox{Hess}(h)(\bar{x},\bar{y})] \rangle , \end{equation}
where angled parenthesis around a function denote the average of that function over one period of the periodic orbit $(\bar{x}(t),\bar{y}(t))$. This condition reduces to (\ref{condition}) when the periodic orbit collapses to a centre. In \cite{Cha15} it was pointed out that the analogue of criterion (\ref{conditiontime}) for limit cycles may return a larger domain of  instability than its counterpart (\ref{condition}), but numerical simulations suggested that the resulting patterns were identical to those obtained in the classical Turing region. Here we follow a different line of reasoning in support of a conjecture that is particular to networks where the underlying planar system is Hamiltonian. Namely, {\em a periodic homogeneous state becomes unstable when the homogeneous centre it surrounds becomes unstable}. We shall not formally verify this conjecture but only sketch out an argument for why one may expect it to be true, at least when the periodic orbit lies sufficiently close the the centre. Perturbations of the homogeneous centre take the form $(x_i, y_i) = ( x^*,y^*) + (\delta x_i, \delta y_i)$  where $(\delta x_i, \delta y_i) = (\delta x, \delta y) + (\delta \bar{x}_i, \delta \bar{y}_i)$ and $(\delta x, \delta y)$ is the piece of the perturbation shared by all the $(x_i, y_i)$. It follows that $(\bar{x}, \bar{y}) = ( x^*,y^*) +  (\delta x, \delta y)$ defines a periodic solution of (\ref{ham}) and therefore $(x_i , y_i) = (\bar{x}, \bar{y})$ for all $i \in V$ is a periodic homogeneous state of the network (\ref{master}). All periodic homogeneous states sufficiently close to the homogeneous centre can be written is this form and their perturbations $(x_i,y_i) = (\bar{x},\bar{y}) + (\delta \bar{x_i}, \delta \bar{y}_i)$ are perturbations of the homogeneous centre by construction. Consequently, if the homogeneous centre is unstable then the periodic homogeneous state is also, and so the conjecture is verified for all periodic orbits sufficiently close to the centre. 

As an illustration we consider the network (\ref{master}) equipped with a simple deformation of the elliptic curve Hamiltonian
\[ h(x,y) = \frac{1}{2} y^2 -\frac{1}{3} x^3 + x + \alpha x y , \quad \alpha > 0 . \]
This yields the system 
\[ \dot{x}_i = y_i + \alpha x_i + D_x \sum_{j=1}^n W_{ij} (x_j- x_i) , \quad \dot{y}_i = x^2_i - 1 - \alpha y_i + D_y \sum_{j=1}^n W_{ij} (y_j - y_i) , \]
and periodic homogeneous solutions persist provided $\alpha $ remains relatively small. The homogeneous centre is specified by
\[ x^* = -(\alpha^2 + \sqrt{\alpha^4 + 4})/2, \quad y^* = \alpha (\alpha^2 + \sqrt{\alpha^4 +4})/2 \]
and the condition (\ref{condition}) for instability becomes
\[ \frac{\alpha^2 (D_x - D_y)^2}{4D_xD_y} > \sqrt{\alpha^4 + 4} . \]   
Without loss of generality we can set $D_y = 1$ and solve this inequality for $D_x$. For example, when $\alpha = 0.1$ we have
\[ 0 < D_x < 0.00124687  \quad \mbox{or} \quad D_x > 802.009, \] 
but can not take the branch $D_x > 802.009$ because this would result in the critical value of the parabola $k(\lambda_j)$ being negative, which is inconsistent with our assumptions. We instead choose a value of $D_x = 0.001$ lying well within the allowed branch and deduce that, for these parameter values, the homogeneous centre is rendered unstable whenever the Laplacian admits an eigenvalue $\lambda$ satisfying
\[ k(\lambda) = (4 \times 0.001) \lambda^2 + 2 \times 0.1 \times (0.001-1) \lambda  + \sqrt{0.1^4 + 4}< 0. \]      
Consider the case where the graph underlying the network (\ref{master}) is complete and unweighted ($W_{ij} = 1$ for all $i,j \in V$) so that nonzero eigenvalues of the Laplacian are $\lambda_j = n/2$ for all $j >1$. Then instability arises whenever $28 < n < 72$ and in this case the condition is realised as a bound on the number of vertices. Numerical simulations confirm that the synchronisation coefficient diverges across the predicted domain (Figure \ref{figure5}). Of course, a connected network of any size and topology can be made to admit the same instability provided the weights $W_{ij}$ are scaled appropriately.

\begin{figure}
\centering
  \includegraphics[width=\textwidth]{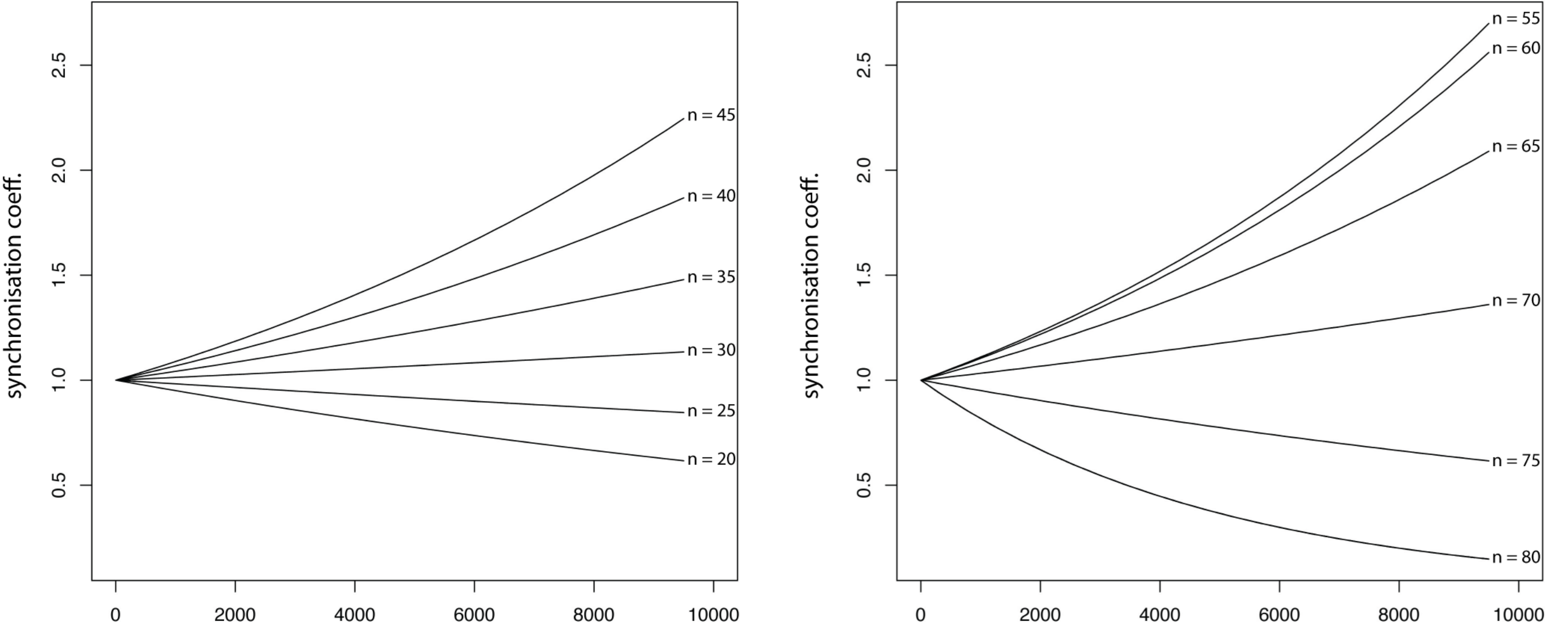}
  \caption{ \label{figure5} Diffusion-induced instability of periodic homogenous states. Turing-like instability is reflected by growth of the synchronisation coefficient, $\sum_{i,j = 1}^n \sqrt{(x_i(t) - x_j(t))^2 + (y_i(t) - y_j(t))^2}$ (normalised by initial conditions), for network sizes lying inside the domain $28 < n < 72$. Exponential decay for network sizes falling outside the unstable region infers synchronisation. Initial conditions $\{(x_j(0),y_j(0) )\}$ were randomly generated to lie on periodic orbits surrounding the centre of the planar system (\ref{ham}). Deformation of the elliptic curve Hamiltonian and parameters were as described in main text. Equations were integrated using the Runge-Kutta method with a time step of $0.01$.}
\end{figure}       

To close this section we describe a simple way of constructing networks of limit cycle oscillators from networks of planar Hamiltonian systems. That conditions for Turing instability appear relatively straightforward to establish in the case of the latter suggests a method for interpolating between these two systems might help overcome some of the challenges encountered with limit cycle oscillators \cite{Cha15}. The method we describe uses the fact that phase portraits of planar Hamiltonian systems are easy to characterise using level sets of the Hamiltonian and often we can describe precisely how these change following a small, dissipative perturbation. More specifically, consider instead of (\ref{master}) networks where the underlying planar system takes the form
\begin{equation} \label{perturbed} \dot{x}= \frac{\partial h (x,y)}{\partial y} + \epsilon f(x,y) , \quad \dot{y} = -\frac{\partial h(x,y)}{\partial x} + \epsilon g(x,y) \end{equation}
with $\epsilon >0$ a small perturbation parameter. The assumption on $h$ is that its level sets contain a family of closed curves. Chapter II.2 of \cite{Chr00} outlines a proof of the Poincar\'{e}-Pontryagin Theorem that says when at least one periodic orbits of the unperturbed Hamiltonian system ($\epsilon = 0$) persists as a limit cycle of the perturbed system ($\epsilon > 0$). In this case we say that the limit cycle is {\em generated} by the corresponding periodic orbit. Roughly, the Poincar\'{e}-Pontryagin criterion is that if the abelian integral
\[ I(h_0) = \int_{h(x,y)=h_0} f(x,y)dy - g(x,y) dx \]         
is not identically zero and satisfies $I(h_0^*) = 0$, $I'(h_0^*) \neq 0$ then there is a unique limit cycle of the perturbed system (\ref{perturbed}) generated by the periodic orbit in the level set $h(x,y) = h^*_0$. This provides a simple criterion for selecting an autonomous planar system whose limit cycles are well characterised by the planar Hamiltonian system (\ref{ham}). Turing instabilities for networks of the dissipative planar system (\ref{perturbed}) are likely to arise nearby to those of the network (\ref{master}).

\section{Concluding remarks}
In this paper we introduced a new class of complex networks (\ref{master}) that consist of diffusively coupled planar Hamiltonian systems. We studied in detail homogeneous versions of these networks, which display special synchronisation and pattern formation properties because planar Hamiltonian systems do not admit attractors. In particular, a number of novel problems arise naturally for this class of networks. These problems include determining which periodic orbits emerge following synchronisation and establishing whether Turing instability of a periodic homogeneous state is equivalent to instability of the homogeneous centre it contains. There are no precise equivalents of these problems for any related networks that have appeared in the literature previously (e.g. \cite{Are08,Nak10, Zan97, Ham99,Vir15,Nig13,Nig15,Nak14,Cha15}). We have made significant progress in some special cases, but obtaining general results remains an open challenge. 

Potential applications for networks (\ref{master}) reside in the fact that we may interpret them as arrangements of coupled oscillators. Of particular relevance are the coupled circadian oscillators that make up the suprachiasmatic nucleus (SCN) of the mammalian brain \cite{Liu97}. In \cite{Liu97} the authors propose that the SCN is made up of a {\em heterogeneous network} of limit cycle oscillators, each individual oscillator with its own intrinsic period. The period of the synchronised state of the SCN is obtained by averaging across these and necessarily only defined with respect to frequency since the network is heterogeneous. Our results on linear networks (\ref{linear}) provide an alternative model where individual oscillators still posses their own intrinsic period, but now on a {\em homogeneous network} where complete synchronisation results in an average period shared across the SCN. Heterogeneous distributions of periods on a homogeneous network are made possible because the underlying system admits a family of periodic orbits. Nonlinear corrections to the linear model may account for non-exact averaging observed in experiments \cite{Liu97}. 

To fully understand networks (\ref{master}) it will be important to move beyond the homogeneous case and consider heterogeneous networks. An accessible starting point would involve taking a heterogeneous distribution of quadratic Hamiltonians $h_i(x,y) = (y^2 + \omega_i^2 x)/2$ for $n$ natural frequencies $\omega_i \in \mathbb{R}$. A comparison between this system and the heterogeneous Kuramoto model \cite{Kur84} can be made by transforming to action-angle variables using $x_i =  \sqrt{2I_i/\omega_i}\cos \theta_i$ and $y_i = -\sqrt{2I_i \omega_i}\sin \theta_i$, which yields
\[ \dot{I}_i = 2 \sum_{j=1}^n \Delta_{ij} \sqrt{ I_iI_j} \left[ D_x \sqrt{\frac{\omega_i}{\omega_j}} \cos \theta_i \cos \theta_j + D_y \sqrt{\frac{\omega_j}{\omega_i}} \sin \theta_i \sin \theta_j  \right] \]
\[ \dot{\theta}_i = \omega_i - \sum_{j=1}^n \Delta_{ij} \sqrt{\frac{I_j}{I_i}} \left[ D_x \sqrt{\frac{\omega_i}{\omega_j}} \sin \theta_i \cos \theta_j - D_y \sqrt{\frac{\omega_j}{\omega_i}} \cos \theta_i \sin \theta_j   \right]   . \]
The second of these equations is reminiscent of the Kuramoto model, but has time-dependent weights $\sqrt{I_j/I_i}$ sitting in front of the coupling terms. This system is therefore an analogue of the Kuramoto model defined on a plastic network \cite{Sel02} where the weights of the network are updated in time according to a very specific learning rule. It may be possible to explore frequency synchronisation on this heterogenous network adapting standard techniques developed for the Kuramoto model.      

\section*{Acknowledgements}
DST is supported by a Research Fellowship from Peterhouse, Cambridge. 
%%\end{document}

\end{document}